\documentclass[12pt]{article}

\usepackage[english]{babel}
\usepackage{dsfont}
\usepackage[letterpaper,top=2cm,bottom=2cm,left=2.5cm,right=2.5cm,marginparwidth=1.75cm]{geometry}

\usepackage{algorithm}
\usepackage{algpseudocode}

\usepackage{dsfont}
\usepackage{amsmath}
\usepackage{tikz}
\usetikzlibrary{fit, positioning, shapes.geometric}
\usetikzlibrary{calc}
\usetikzlibrary{matrix}

\usepackage{polyomino}

\usepackage{nicematrix}
\usepackage{mathtools}

\usepackage{quoting}
\quotingsetup{vskip=0pt}

\usepackage{amsmath,amssymb,bm,amsthm}
\newtheorem{proposition}{Proposition}

\newtheorem{theorem}{Theorem}
\newtheorem{condition}{Condition}

\newtheorem{definition}{Definition}

\usepackage{multicol}
\usepackage{graphicx}
\usepackage{tikz}
\usetikzlibrary{backgrounds,matrix,fit}

\usepackage{pgfplots}
\pgfplotsset{compat=newest}
\pgfplotsset{plot coordinates/math parser=false}

\usepackage{amsmath}
\usepackage{mathtools}

\allowdisplaybreaks

\def\bs{\ensuremath\boldsymbol}

\usepackage{graphicx}
\usepackage{enumerate}
\usepackage[colorlinks=true, allcolors=blue]{hyperref}
\usepackage{authblk}
\usepackage{nicematrix}
\usepackage{algorithm}
\usepackage{algpseudocode}

\title{Infinite families of graphs and stable completion of arbitrary matrices, Part I}
\author{Augustin Cosse\\
\textcolor{blue}{augustin.cosse@univ-littoral.fr}}
\affil{Universit\'e du Littoral C\^ote d'Opale}

\begin{document}
\maketitle

\begin{abstract}
We study deterministic constructions of graphs for which the unique completion of low rank matrices is generically possible regardless of the values of the entries.  We relate the completability to the presence of some patterns (particular unions of self-avoiding walks) in the subgraph of the lattice graph generated from the support of the bi-adjacency matrix.  The construction makes it possible to design infinite families of graphs on which exact and stable completion is possible for every fixed rank matrix through the sum-of-squares hierarchy.

\end{abstract}

\section{Introduction}

Given a matrix $\bs X_0\in \mathcal{M}(m\times n, r) \subset \mathbb{R}^{m\times n}$,  where $\mathcal{M}(m\times n; r)$ denotes the manifold of rank-$r$ matrices,  the completion problem consists in recovering $\bs X_0$ given a subset $\Omega \subset E = [m]\times [n]$ of its entries.  The known and missing entries of the matrix are usually represented through the bipartite graph $G = (V, W, E)$ with $V = [m],  W = [n]$ corresponding to the row and column indices of $\bs X_0$.  In its general form, the low rank matrix completion problem finds numerous applications from collaborative filtering~\cite{goldberg1992using} to vision and control~\cite{chen2004recovering}, including genomic data integration~\cite{cai2016structured}. 

In the particular case of Gram (i.e. positive semidefinite matrices) of the form $X = PP^\intercal$,  the columns of $P$ can be regarded as points in $\mathbb{R}^r$ and the matrix $X$ can be regarded as specifying a subset of pairwise distances between points in $\mathbb{R}^r$.  The problem then consists in recovering the positions of all the points.  Applications of this \emph{geometric} formulation sometimes known as \emph{graph realization} can be found in sensor localization~\cite{singer2008remark, biswas2006semidefinite},  structural biology (i.e. molecular conformation, NMR spectroscopy and the molecule problem)~\cite{hadani2011representation}.

There exist a large body of work on the number of measurements and/or structure of the graph garanteeing the unique (and stable) recovery of the matrix.  Those can be roughly separated into two main categories which are discussed in more details in sections~\ref{pseudoRandomSchemes} and~\ref{discussionAlgebraicMC}. 
\begin{itemize}
\item Random or pseudo-random (i.e. expander based) sampling and/or incoherence assumptions.  The combination of random sampling and incoherence appears in classical compressed sensing references such as~\cite{candes2010power, candes2008exact, recht2011simpler} as well as~\cite{keshavan2009matrix, keshavan2010matrix} or even~\cite{bhojanapalli2014universal, burnwal2020deterministic, burnwal2019construction, butnwal2019some, burnwal2020exact} (Ramanujan + incoherence) 
\item Deterministic/algebraic constructions and matroid/rigidity theory.  The classical references on deterministic sampling patterns include~\cite{kiraly2015algebraic}, ~\cite{singer2010uniqueness} and~\cite{pimentel2016characterization} as well as the litterature on rigidity theory and (see e.g.~\cite{cucuringu2012graph} and the additional references discussed in section~\ref{discussionAlgebraicMC} below)
\end{itemize}

\section{Main contributions}

What motivates this work is the following Theorem from~\cite{kiraly2015algebraic}
\begin{theorem}\label{theoremOnlystructure01}
Let $E\subset [m]\times [n]$ be a subset of positions, $(k, \ell)\in [m]\times [n]\setminus E$ be arbitrary and let $\bs X_0\in \mathbb{R}^{m\times n}$ be a generic $(m\times n)$ matrix of rank $r$. Whether the entry $(X_0)_{k, \ell}$ at a position $(k, \ell)$ is uniquely completable from the $(\bs X_0)_{ij}$, $(i,j)\in E$, depends only on the position $(k, \ell)$, the true rank $r$, and the observed positions $E$ (and not on the value of the entries themselves)
\end{theorem}
Given this theorem,  and in particular the result that was derived in~\cite{cosse2021stable} it seems natural to wonder whether the incoherence assumptions in~\cite{bhojanapalli2014universal} or~\cite{candes2010power} or the use of the $\gamma_2$-norm in the estimate from~\cite{heiman2014deterministic} are really needed.  In view of~\cite{cosse2021stable} it seems as though a convex recovery result free of any assumption on the distribution of the singular vectors could be achieved by only a slightly more complex semidefinite program.  In this paper we rely on the support of the bi-adjacency matrix to derive a completion result based on the existence of a sequence of subgraphs $\left(K_0\right)_{2,2}^- \rightarrow \left(K_1\right)_{2,2}^{2-} \rightarrow \ldots \left(K_\ell\right)_{2,2}^{2-}$ where $(K_{\ell}^{k-})$ denote complete bipartite graphs with $k$ edges missing.  Building on this idea, one should reasonably expect to be able to derive a similar result for arbitrary rank matrices. In particular, we can expect the existence of unique tractable rank-2 completion on the bi-adjacency matrices of graphs exhibiting a chains of the form 
\begin{align}
\left\{(K_a)_{3,3}^{-}\right\}_a\rightarrow \left\{(K_a)_{3,3}^{2-}\right\}_a \rightarrow \ldots \rightarrow \left\{(K_{3,3})_a^{9-}\right\}_a\label{bipartitePath}\end{align}
such that every set of $K_{3,3}^{\ell-}$ graphs can be completed from the knowledge of the previously completed $K_{3,3}^{m-}$ with $m<\ell$. In this note we give a simple sufficient condition on the adjacency matrix for such paths to exist within the graph. The condition is related to the existence of self-avoiding walks/cycles inside the subgraph of the lattice graph generated by the support of the bi-adjacency matrix.

\subsection{Main results}

\subsubsection{Mathematical Preliminaries}

In order to state the main results, we will need a few notions related to self-avoiding walks. We will also very briefly recall the key ideas behind the sum-of-squares hierarchy (see~\cite{laurent2008sums} for details). We will then start by addressing the rank 2 case (section~\ref{rank2caseSection}) and finally we will discuss an extension to the general low rank completion setting in section~\ref{generalCase}. 

\begin{definition}[see e.g.~\cite{polypoly2009, tullekenpolyominoes2}, also~\cite{van2015statistical}]
Let $\mathbb{Z}^d$ denote the hypercubic lattice.  A $n$-step self-avoiding walk $\partial\Omega\subset \mathbb{Z}^d$ is a sequence $\omega$ of distinct points $\omega(0), \omega(1), \ldots, \omega(n)$ in $\mathbb{Z}^d$ such that each point is a nearest neighbor of its predecessor, i.e. $\omega(k+1) - \omega(k)\in \left\{x\in \mathbb{Z}^d\;:\; \|x\|_1 = 1\right\}$.  If the end point of a self-avoiding walk, $\omega(n)$ is adjacent to the origin $\omega(0)$ (i.e. $\omega(n)\sim\omega(0)$), an additional step joining the end-point to the origin will produce a \emph{self-avoiding circuit}.
\end{definition}

Given a path $\omega$, we will use the notation $\omega[v, w]$ to denote the subwalk from time $v$ to $w$. 

Consider the bi-adjacency matrix $A$ associated to the graph $G([m],[n], E)$, i.e. $A\in \left\{0,1\right\}^{m\times n}$.  We define the sub-region of the square lattice
\begin{align}
\Lambda_{m,n} = \left\{(i,j)\in \mathbb{Z}^2\;|\; 1\leq i\leq m, \;1\leq j\leq n\right\}
\end{align}
and we define the support 
\begin{align}
S = \left\{(i,j)\in \Lambda_{m,n}\; |\; A_{ij}=1\right\}
\end{align}
We consider the \emph{subgraph} $\mathcal{G}$  \emph{of the lattice graph} generated from $S$.

\subsubsection{Sums-of-squares}

Given the system of polynomial equations 
\begin{align}
S = \left\{X_{ij} - (X_0)_{ij}\right\}_{ij\in E}\cup \left\{\text{det}(\tilde{X})_{I,J}^{E}(X)\right\}_{|I|, |J| = r}\label{polySystem01}
\end{align} 
where $\text{det}(\tilde{X})_{I,J}( X)$ denotes the set of polynomial equations (in the entries of $X_{E^c}$) generated from the minors associated to the matrix $\tilde{X} = (X_0)_E + (X - X_E)$ where we have replaced the entries in $E$ by their value in $X_0$, we use $\mathcal{I}$ to denote the ideal associated to the system,  i.e.
\begin{align}
\mathcal{I} = \left\{\sum_{j} h_j(X) p_j(X)\;|\;p_j(X)\in \mathbb{R}[X], \; h_j(X)\in S\right\}
\end{align}
It is known (see for example~\cite{cosse2021stable}) that as soon as one can show that the system~\eqref{polySystem01} has a unique solution, a stability estimate can be derived by means of a semidefinite program of size $O((mn)^t)$ provided that one can show 
\begin{align}
X_{ij} - (X_0)_{ij} \in \mathcal{I}_t
\end{align}
where $\mathcal{I}_t$ is the truncation of $\mathcal{I}$ to degree $t$, that is to say to polynomials $p_j(X)$ with $\text{\upshape deg}(p_j(X))\leq t - \max(\text{\upshape deg}(h_j(X)))$.  The semidefinite program corresponds to the minimization of the Trace norm over the set of (truncated) moments matrices (We will not go into the details of those ideas here and refer the interested reader to the monograph~\cite{laurent2008sums}). In this work we focus on showing that, under particular patterns in the subgraph of the lattice graph associated to $E = \text{\upshape supp}(X_0)$, it is possible to show that the system~\eqref{polySystem01} admits a single solution and by relying on the ideas developed in~\cite{cosse2021stable} that all the monomials $X_{ij} - (X_0)_{ij}$ can be expressed from the ideal $\mathcal{I}_t$ for an appropriate choice of $t$.

\subsubsection{\label{rank2caseSection}Rank 2 case}

\begin{definition}[rectangle]
Given a path $\omega = (\omega(0), \ldots, \omega(n))\subset \mathbb{Z}^2$ on the two-dimensional lattice, we say that a subpath of $\omega$ \emph{forms a rectangle} if there exist indices $i_1<i_2<i_3<i_4<i_5$ and integers $a<b, c<d$ such that
\begin{align*}
\omega(i_1) = (a, c)\\
\omega(i_2) = (b, c)\\
\omega(i_3) = (b,d)\\
\omega(i_4) = (a, d)\\
\omega(i_5) = (a, c)
\end{align*}
and each segment $\omega[i_j, i_{j+1}]$ is a straight horizontal and vertical lattice path. 
\end{definition}

\begin{definition}[\label{oneRemovableVertex}1-removable vertex]
Given a self avoiding circuit $C_0$, for any concave/reflex vertex $\omega(v)$, we consider the intersection between each of the two edges making up the vertex and the part of the cycle facing those edges $\omega\pm (L_x, 0)$ and $\omega(v)\pm(0, L_y)$ we denote those intersections as $\omega(v^x)$ and $\omega(v^y)$.  We call the vertex $1$-removable if at least of the following conditions is satisfied
\begin{itemize}
\item $\omega[v, v^x]\cup \left\{\omega(v), \omega(v)+(1,0), \ldots \omega(v)+(L_x,0)\right\}$ forms a rectangle\\ and $\left\{\omega(v)+(0,K), \ldots, \omega(v)+(L_x, K)\right\}\subset \Omega$ for some $K\in \mathbb{Z}$
\item $\omega[v^x, v]\cup \left\{\omega(v), \omega(v)+(1,0), \ldots \omega(v)+(L_x,0)\right\}$ forms a rectangle\\ and $\left\{\omega(v)+(0,K), \ldots, \omega(v)+(L_x, K)\right\}\subset \Omega$ for some $K\in \mathbb{Z}$
\item $\omega[v, v^y]\cup \left\{\omega(v), \omega(v)+(0,1), \ldots \omega(v)+(0,L_y)\right\}$ forms a rectangle\\ and $\left\{\omega(v)+(K,0), \ldots, \omega(v)+(K,L_y)\right\}\subset \Omega$ for some $K\in \mathbb{Z}$
\item $\omega[v^y, v]\cup \left\{\omega(v), \omega(v)+(0,1), \ldots \omega(v)+(0,L_y)\right\}$ forms a rectangle\\
and $\left\{\omega(v)+(K,0), \ldots, \omega(v)+(K,L_y)\right\}\subset \Omega$ for some $K\in \mathbb{Z}$
\end{itemize}
\end{definition}

\begin{definition}[L-removable vertex]
Let us denote as $C_1$ the cycle obained by ``removing" all $1$-removable reflex vertices from $C_0$ and closing the correspondig subcycles. Then we call $2$-removable with respect to $C_0$ the $1$ removable reflex vertices of $C_1$. Similarly, if we denote as $C_\ell$ the cycle obtained after $\ell$ steps corresponding to successively removing each of the removable vertices, then a reflex vertex is called $L$-removable with respect to $C_0$ if it $1$-removable with respect to $C_{L-1}$.  
\end{definition}

\begin{proposition}[\label{propositionCycleSelfAvoiding01}Rank 2 case]
Let $G = ([m], [n],E)$ denote any bipartite graph isomorphic to a graph $G^*$ for which the associated subgraph $\mathcal{G}^*\subset \Lambda_{m,n}$  of the lattice graph contains a cycle in which the reflex vertices are at most \emph{$L$-removable} as shown in Fig.~\ref{closableVSNonClosable}. Then the associated rank-2 matrix can be completed exactly and stably through the level $O(L\vee N_r)$ where $N_r$ is the total number of reflex vertices of the cycle of the sum-of-squares hierarchy
\end{proposition}

\begin{proof}
Given the self-avoiding cycle $\omega$, we define the embedded polygonal curve $\Gamma$ as
\begin{align}
\Gamma = \bigcup_{t=0}^{L_1} [\omega(t), \omega(t+1)]
\end{align}
By the Jordan curve theorem $\mathbb{Z}^2\setminus \Gamma$ has exactly two components, one bounded (the interior) and one unbounded (the exterior).  We use $\text{Int}(\omega)$ to denote the closure of the bounded component of $\mathbb{Z}_2\setminus \Gamma$ and accordingly $\text{Ext}(\omega) = \Lambda_{m,n}\setminus \text{Int}(\omega)$.  We start by considering the interior.  For every $1$-removable vetex,  we consider any of the four rectangles described in Definition~\ref{oneRemovableVertex}.  Let us take any position $(i,j)$ inside this rectangle. From Definition~\ref{oneRemovableVertex}, there exist $K_1, K_2$ and $M_1, M_2$ such that $(i, j-K_1), (i, j+K_2), (i+M_1, j), (i+M_2, j), (i+K_1, j+M_1), (i+K_2, j+M_1), (i+K_1, j+M_2), (i+K_2, j+M_2)\in\omega$.  Let $I = (i, i+K_1, i+K_2)$ and $J = (j, j+M_1, j+M_2)$ and let us use
$(X_0)_{I,J}^{(i,j)}$ to denote the submatrix of $X_0$ where we replace the entry at position $(i,j)$ by the variable $x$ and the monomial $X_{ij} - (X_0)_{ij}$ belongs to the ideal generated by the minors obtained after substitution of the known entries, i.e.  $X_{ij} - (X_0)_{ij}\in I_M$.  Then the equation $\text{det} (X_0)_{I,J}^{(i,j)}(x) = 0$ generically has a single solution that corresponds to the value of $x$.  Proceeding like this, one can ``close the rectangle" and uniquely determine the values of all the entries in this rectangle.  Once all the $1$-removable vertices have been ``closed", we proceed with the $2$-removable vertices and so on.  For 2-removable vertices, one can again apply the procedure to any position $(k, \ell)$ located inside the rectangle with the minor $(X_0)_{I,J}$ having at most two missing entries. If there is a single missing entry, we apply the same procedure as for 1-removable vertices. In the case where there $2$ missing entries, consider the polynomial $\text{det} (X_0)_{I,J}^{(i,j)(k, \ell)}(x_1, x_2) = 0$. Since this polynomial vanishes at $x_1 = (x_0)_1$ one can always write it as 
\begin{align}
\text{det} (X_0)_{I,J}^{(i,j)(k, \ell)}(x_1, x_2)  = (x_1 - (x_0)_1) p(x_2)
\end{align}
where $p(x_2)$ is of the form $p(x_2) = \alpha x_2 + \beta$.  In particular one can always eliminate the $x_1x_2$ and $x_1$ terms from $\text{det} (X_0)_{I,J}^{(i,j)(k, \ell)}(x_1, x_2) $ to express a polynomial $q(x_2) = \alpha'x_2 + \beta'$ as 
\begin{align}
\alpha'x_2 + \beta' = \text{det} (X_0)_{I,J}^{(i,j)(k, \ell)}(x_1, x_2)  - C_1(x_1 - (x_0)_1)x_2 - C_2(x_1 - (x_0)_1)
\end{align}
for appropriate constants $C_1, C_2$.  Since both $\text{det} (X_0)_{I,J}^{(i,j)(k, \ell)}(x_1, x_2)$ and $(x_1 - (x_0)_1)$ are in $I_M$, so is $\alpha'x_2 + \beta'$. We can proceed like this for all the remaining vertices obtaining a decomposition for $X_{ij} - (X_0)_{ij}$ of degree at most $L$ from elements of $I_M$. For all the missing entries inside the cycle we can thus write $X_{ij} - (X_0)_{ij}\in I_M^{(L)}$ where $I_M^{(L)}$ is the truncation of the ideal generated by the degree $3$ minors to degree $L$.  Combined with the results of~\cite{cosse2021stable} this implies that the interior of the cycle can be completed exactly and stably from the level-$L$ of the sos hierarchy combined with the trace norm.  

For the exterior of the rectilinear polygon, since the cycle spans the whole set of row and column indices, either it is given by the boundary of $\Lambda_{m,n}$ or we can iteratively complete each of the rectangles whose sides are defined by the subpaths located between two turns as shown in Fig,~\ref{completionExterior}.  For any position $(i,j)$ on the exterior and inside those rectangles, since the cycle is self-avoiding,  there must exist $K_1, K_2$ and $M_1, M_2$ such that $(i+K_1, j), (i+K_2, j), (i, j+M_1), (i, j+M_2), (i+K_1, j+K_1), (i+K_1, j+K_2), (i+K_2, j+K_1), (i+K_2, j+K_2) \in \text{Int}(\omega)$. As before we can then express the corresponding monomials $X_{ij} - (X_{0})_{ij}$ from the ideal generated by the rank-2 minors. 
\end{proof}

\begin{figure}
\centering 

\begin{minipage}{.45\linewidth}
\centering

\begin{tikzpicture}[scale=0.5]

\draw[fill=blue!30] (0,9) rectangle (1,10);
\draw[fill=blue!30] (0,8) rectangle (1,9);
\draw[fill=blue!30] (1,8) rectangle (2,9);
\draw[fill=blue!30] (2,8) rectangle (3,9);
\draw[fill=blue!30] (2,7) rectangle (3,8);
\draw[fill=blue!30] (3,7) rectangle (4,8);
\draw[fill=blue!30] (4,7) rectangle (5,8);
\draw[fill=blue!30] (5,7) rectangle (6,8);
\draw[fill=blue!30] (5,6) rectangle (6,7);
\draw[fill=blue!30] (6,6) rectangle (7,7);
\draw[fill=blue!30] (6,5) rectangle (7,6);
\draw[fill=blue!30] (6,4) rectangle (7,5);
\draw[fill=blue!30] (7,4) rectangle (8,5);
\draw[fill=blue!30] (7,3) rectangle (8,4);
\draw[fill=blue!30] (7,2) rectangle (8,3);
\draw[fill=blue!30] (8,2) rectangle (9,3);
\draw[fill=blue!30] (8,1) rectangle (9,2);
\draw[fill=blue!30] (9,1) rectangle (10,2);

\draw[fill=blue!30] (9,0) rectangle (10,1);

\draw[fill=blue!15] (0,7) rectangle (2,8);
\draw[fill=blue!15] (2,6) rectangle (5,7);
\draw[fill=blue!15] (1,9) rectangle (3,10);
\draw[fill=blue!15] (3,8) rectangle (6,9);

\draw[fill=blue!15] (5,4) rectangle (6,6);
\draw[fill=blue!15] (7,5) rectangle (8,7);
\draw[fill=blue!15] (6,7) rectangle (7,8);

\draw[fill=blue!15] (8,3) rectangle (9,5);

\draw[fill=blue!15] (6,2) rectangle (7,4);
\draw[fill=blue!15] (7,1) rectangle (8,2);
\draw[fill=blue!15] (8,0) rectangle (9,1);

\draw[fill=blue!15] (9,2) rectangle (10,3);

\foreach \x in {0,...,10} \draw[gray!50] (\x,0) -- (\x,10);
\foreach \y in {0,...,10} \draw[gray!50] (0,\y) -- (10,\y);

\draw[line width=1.5pt] 
    (0,9) -- (0,8) -- (1,8) -- (2,8) -- (2,7) -- (3,7) -- (4,7) -- (5,7) -- (5,6) -- (6,6) -- (6,5) -- (6,4) -- (7,4) -- (7,3) -- (7,2)--(8,2) -- (8,1) -- (9,1) -- (9,0) -- (10,0) --  (10,2)--(9,2) --(9,3) -- (8,3) -- (8,4) --  (8,5) --(7,5) --  (7,6) -- (7,7) -- (6,7) -- (6,8)--  (5,8) -- (4,8) -- (3,8) -- (3,9) -- (2,9) -- (1,9) -- (1,10) -- (0,10) -- (0,9) ;

\end{tikzpicture}
\end{minipage}
\begin{minipage}{.45\linewidth}
\centering
\begin{tikzpicture}[x={(1cm,0cm)}, y={(0.5cm,0.5cm)}, z={(0cm,1cm)}]

\def\s{0.5}       
\def\op{0.35}     
\def\N{5}         

\foreach \i in {0,...,5} {
    \foreach \j in {0,...,5} {
        \draw[gray!40] (\i*\s,\j*\s,0) -- (\i*\s,\j*\s,\N*\s);
        \draw[gray!40] (\i*\s,0,\j*\s) -- (\i*\s,\N*\s,\j*\s);
        \draw[gray!40] (0,\i*\s,\j*\s) -- (\N*\s,\i*\s,\j*\s);
    }
}

\newcommand{\cube}[5]{%
  \coordinate (A) at (#1,       #2,       #3);
  \coordinate (B) at (#1+#4,    #2,       #3);
  \coordinate (C) at (#1+#4,    #2+#4,    #3);
  \coordinate (D) at (#1,       #2+#4,    #3);
  \coordinate (E) at (#1,       #2,       #3+#4);
  \coordinate (F) at (#1+#4,    #2,       #3+#4);
  \coordinate (G) at (#1+#4,    #2+#4,    #3+#4);
  \coordinate (H) at (#1,       #2+#4,    #3+#4);

  \fill[blue!50, opacity=#5, opacity=#5] (A)--(B)--(C)--(D)--cycle; 
  \fill[blue!50, opacity=#5, opacity=#5] (E)--(F)--(G)--(H)--cycle; 
  \fill[blue!50, opacity=#5, opacity=#5] (A)--(B)--(F)--(E)--cycle; 
  \fill[blue!50, opacity=#5, opacity=#5] (B)--(C)--(G)--(F)--cycle; 
  \fill[blue!50, opacity=#5, opacity=#5] (C)--(D)--(H)--(G)--cycle; 
  \fill[blue!50, opacity=#5, opacity=#5] (D)--(A)--(E)--(H)--cycle; 
  \draw[gray!40] (A)--(B)--(C)--(D)--cycle;
  \draw[gray!40] (E)--(F)--(G)--(H)--cycle;
  \draw[gray!40] (A)--(E);
  \draw[gray!40] (B)--(F);
  \draw[gray!40] (C)--(G);
  \draw[gray!40] (D)--(H);
}


\foreach \z in {0,1,2,3,4} {
  \pgfmathsetmacro{\zflip}{\N*\s - \z*\s - \s}
  \cube{0}{4*\s}{\zflip}{\s}{\op}
}

\foreach \x in {0,1,2,3,4} {
  \pgfmathsetmacro{\zflip}{\N*\s - 4*\s - \s} 
  \cube{\x*\s}{4*\s}{\zflip}{\s}{\op}
}

\foreach \y in {4,3,2,1,0} {
  \pgfmathsetmacro{\zflip}{\N*\s - 4*\s - \s} 
  \cube{4*\s}{\y*\s}{\zflip}{\s}{\op}
}

\draw[line width=1.2pt, black] (5*\s,0,0) -- (5*\s,5*\s,0);
\draw[line width=1.2pt, black] (\s,5*\s,\s) -- (5*\s,5*\s,\s);
\draw[line width=1.2pt, black] (5*\s,5*\s,0) -- (5*\s,5*\s,\s);

\draw[line width=1.2pt, black] (\s,5*\s,\s) -- (\s,5*\s,5*\s);

\draw[line width=1.2pt, black] (0,4*\s,0) -- (0,4*\s,5*\s);

\draw[line width=1.2pt, black] (0,4*\s,5*\s) -- (0,5*\s,5*\s);

\draw[line width=1.2pt, black] (0,5*\s,5*\s) -- (\s,5*\s,5*\s);

\draw[line width=1.2pt, black] (0,4*\s,0) -- (3*\s,4*\s,0);

\draw[line width=1.2pt, black] (4*\s,0,\s) -- (4*\s,2*\s,\s);

\draw[line width=1.2pt, black] (4*\s,0,0) -- (5*\s,0,0);

\draw[line width=1.2pt, black] (4*\s,0,\s) -- (4*\s,0,0);

\end{tikzpicture}
\end{minipage}

\caption{\label{completionExterior}An example of a polyomino type cycle covering the full set of column and row indices and the corresponding first step (degree 3) of the completion of the exterior. The (white) cells are assumed to be of size $3$ at least. The same idea carries over to the tensor setting replacing the notion of bipartite graphs by the more general notion of hypergraphs illustrated in Fig.~\ref{illustrationHypergraph01}. }
\end{figure}
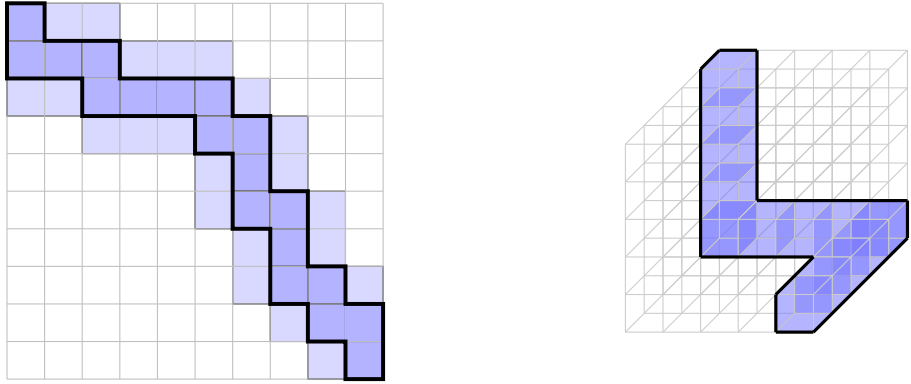

\begin{figure}
\centering
\begin{minipage}{.45\linewidth}
\begin{tikzpicture}
    \node (v1) at (0,3) {};
    \node (v2) at (0,2) {};
    \node (v3) at (0,1) {};
    
    \node (v4) at (3,3) {};
    \node (v5) at (3,2) {};
    \node (v6) at (3,1) {};
    
    \node (v7) at (6,3) {};
    \node (v8) at (6,2) {};
    \node (v9) at (6,1) {};

    \begin{scope}[fill opacity=0.4]
        \fill[yellow!70]
            ($(v1)+(-0.2,0)$)
            to[out=90,in=180] ($(v5)+(0,0.3)$)
            to[out=0,in=90] ($(v9)+(0.2,0)$)
            to[out=270,in=0] ($(v5)+(0,-0.3)$)
            to[out=180,in=270] ($(v1)+(-0.2,0)$);
        
        \fill[blue!70]
            ($(v2)+(-0.2,0)$)
            to[out=90,in=180] ($(v4)+(0,0.3)$)
            to[out=0,in=90] ($(v8)+(0.2,0)$)
            to[out=270,in=0] ($(v4)+(0,-0.3)$)
            to[out=180,in=270] ($(v2)+(-0.2,0)$);
        
        \fill[green!70]
            ($(v3)+(-0.2,0)$)
            to[out=90,in=180] ($(v6)+(0,0.3)$)
            to[out=0,in=90] ($(v7)+(0.2,0)$)
            to[out=270,in=0] ($(v6)+(0,-0.3)$)
            to[out=180,in=270] ($(v3)+(-0.2,0)$);
        
        \fill[red!70]
            ($(v1)+(-0.2,0)$)
            to[out=90,in=180] ($(v4)+(0,0.3)$)
            to[out=0,in=90] ($(v9)+(0.2,0)$)
            to[out=270,in=0] ($(v4)+(0,-0.3)$)
            to[out=180,in=270] ($(v1)+(-0.2,0)$);
    \end{scope}

    \foreach \i in {1,...,9} {
        \fill (v\i) circle (0.1);
    }

    \node[below] at (v1) {$v_1$};
    \node[below] at (v2) {$v_2$};
    \node[below] at (v3) {$v_3$};
    \node[below] at (v4) {$v_4$};
    \node[below] at (v5) {$v_5$};
    \node[below] at (v6) {$v_6$};
    \node[below] at (v7) {$v_7$};
    \node[below] at (v8) {$v_8$};
    \node[below] at (v9) {$v_9$};

    \node at (2,3.5) {$e_1$};
    \node at (2,2.5) {$e_2$};
    \node at (3.5,1.5) {$e_3$};
    \node at (3,3.2) {$e_4$};
\end{tikzpicture}
\end{minipage}
\begin{minipage}{.45\linewidth}
\begin{tikzpicture}[x=(15:.5cm), y=(90:.5cm), z=(330:.5cm), >=stealth]
\draw (0, 0, 0) -- (0, 0, 10) (3, 0, 0) -- (3, 0, 10);
\foreach \z in {0, 5, 10} \foreach \x in {0,...,2}
  \foreach \y [evaluate={\b=random(0, 0);}] in {0,...,2}
    \filldraw [fill=white] (\x, \y, \z) -- (\x+1, \y, \z) -- (\x+1, \y+1, \z) --
      (\x, \y+1, \z) -- cycle (\x+.5, \y+.5, \z) node [yslant=tan(15)] {};
\draw [dashed] (0, 3, 0) -- (0, 3, 10) (3, 3, 0) -- (3, 3, 10);
\draw [->] (0, 3.5, 0)  -- (3, 3.5, 0)   node [near end, above left] {};
\draw [->] (-.5, 3, 0)  -- (-.5, 0, 0)   node [midway, left] {};
\draw [->] (3, 3.5, 10) -- (3, 3.5, 2.5) node [near end, above right] {};
\filldraw [fill=white] (0.5,0.5,10) node [yslant=tan(15)]  {1};
\filldraw [fill=white] (0.5,2.5,0) node [yslant=tan(15)]  {1};
\filldraw [fill=white] (1.5,2.5,0) node [yslant=tan(15)]  {1};
\filldraw [fill=white] (0.5,1.5,5) node [yslant=tan(15)]  {1};
\filldraw [fill=white] (1.5,1.5,5) node [yslant=tan(15)]  {0};
\filldraw [fill=white] (1.5,2.5,5) node [yslant=tan(15)]  {0};
\filldraw [fill=white] (2.5,1.5,5) node [yslant=tan(15)]  {0};
\filldraw [fill=white] (2.5,2.5,5) node [yslant=tan(15)]  {0};

\filldraw [fill=white] (0.5,2.5,5) node [yslant=tan(15)]  {0};
\filldraw [fill=white] (0.5,0.5,5) node [yslant=tan(15)]  {0};
\filldraw [fill=white] (1.5,0.5,5) node [yslant=tan(15)]  {0};
\filldraw [fill=white] (2.5,0.5,5) node [yslant=tan(15)]  {0};

\filldraw [fill=white] (1.5,0.5,10) node [yslant=tan(15)]  {0};
\filldraw [fill=white] (2.5,0.5,10) node [yslant=tan(15)]  {0};

\filldraw [fill=white] (0.5,1.5,10) node [yslant=tan(15)]  {0};
\filldraw [fill=white] (1.5,1.5,10) node [yslant=tan(15)]  {0};
\filldraw [fill=white] (2.5,1.5,10) node [yslant=tan(15)]  {0};
\filldraw [fill=white] (0.5,2.5,10) node [yslant=tan(15)]  {0};
\filldraw [fill=white] (1.5,2.5,10) node [yslant=tan(15)]  {0};
\filldraw [fill=white] (2.5,2.5,10) node [yslant=tan(15)]  {0};

\filldraw [fill=white] (0.5,0.5,0) node [yslant=tan(15)]  {0};
\filldraw [fill=white] (1.5,0.5,0) node [yslant=tan(15)]  {0};
\filldraw [fill=white] (2.5,0.5,0) node [yslant=tan(15)]  {0};
\filldraw [fill=white] (0.5,1.5,0) node [yslant=tan(15)]  {0};
\filldraw [fill=white] (1.5,1.5,0) node [yslant=tan(15)]  {0};
\filldraw [fill=white] (2.5,1.5,0) node [yslant=tan(15)]  {0};

\filldraw [fill=white] (2.5,2.5,0) node [yslant=tan(15)]  {0};

\end{tikzpicture}
\end{minipage}
\caption{\label{illustrationHypergraph01}An example of a $3$-partite hypergraph and the associated tensor mask. }
\end{figure}
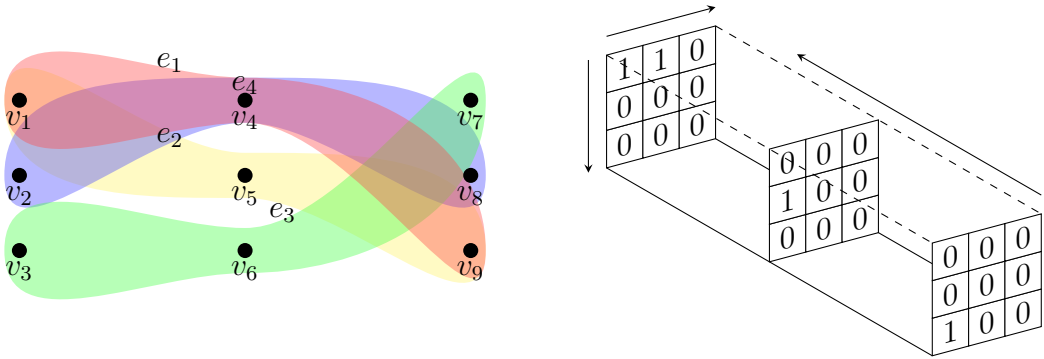

\begin{figure}
\centering 

\begin{minipage}{.45\linewidth}
\centering
\begin{tikzpicture}[scale=0.5]


\draw[fill=blue!35] (2,5) rectangle (3,6);
\draw[fill=blue!25] (1,4) rectangle (4,5);
\draw[fill=blue!15] (1,0) rectangle (5,1);
\draw[fill=blue!5] (0,1) rectangle (6,4);

\foreach \x in {0,...,6} \draw[gray!50] (\x,0) -- (\x,6);
\foreach \y in {0,...,6} \draw[gray!50] (0,\y) -- (6,\y);

\draw[line width=1.5pt] 
    (0,1) -- (1,1)--(1,0) -- (2,0) --(3,0) -- (4,0) -- (5,0) --(5,1)-- (6,1) -- (6,2) -- (6,3) -- (6,4) -- (5,4) -- (4,4) --(4,5) -- (3,5) --  (3,6) -- (2,6) --  (2,5) -- (1,5)-- (1,4) -- (0,4) -- (0,3)-- (0,2) -- (0,1);

\end{tikzpicture}
\end{minipage}
\begin{minipage}{.45\linewidth}
\centering
\begin{tikzpicture}[scale=0.5]


\foreach \x in {0,...,6} \draw[gray!50] (\x,0) -- (\x,6);
\foreach \y in {0,...,6} \draw[gray!50] (0,\y) -- (6,\y);

\draw[line width=1.5pt] 
    (0,1) -- (1,1)--(2,1) -- (2,0) --(3,0) -- (4,0) -- (5,0) --(5,1)-- (6,1) -- (6,2) -- (6,3) -- (6,4) -- (5,4) -- (4,4) --(4,6) --  (3,6) -- (2,6)-- (1,6)-- (1,4) -- (0,4) -- (0,3)-- (0,2) -- (0,1);

\end{tikzpicture}
\end{minipage}

\caption{\label{closableVSNonClosable}Examples of removable and non removable vertices (see Definition~\ref{oneRemovableVertex})}
\end{figure}
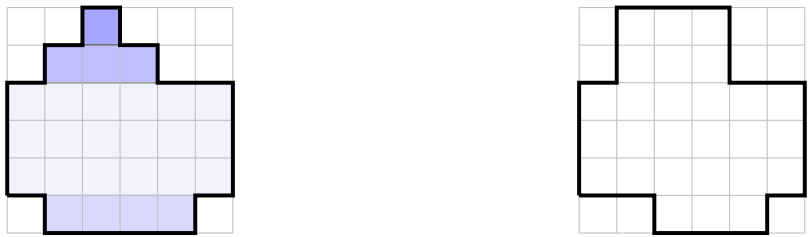

%
%
%
%

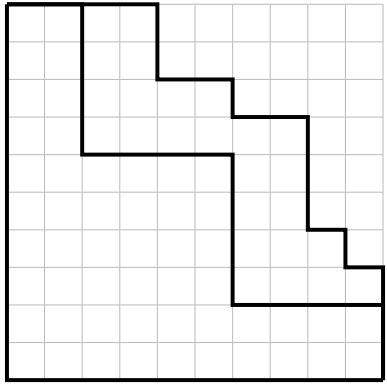
\begin{figure}
\begin{center}
\begin{tikzpicture}[scale=0.5]

\foreach \x in {0,...,10} \draw[gray!50] (\x,0) -- (\x,10);
\foreach \y in {0,...,10} \draw[gray!50] (0,\y) -- (10,\y);

\draw[line width=1.5pt] 
    (0,10) -- (0,0) -- (10,0) ;

\draw[line width=1.5pt] 
    (0,10) -- (2,10) -- (2, 6) -- (6, 6) -- (6,2)-- (10,2) -- (10,0);

\draw[line width=1.5pt] 
    (0,10) -- (4,10) -- (4, 8) -- (6, 8) -- (6,7) -- (8,7) -- (8,4)-- (9,4)-- (9,3)-- (10,3) -- (10,0);

\end{tikzpicture}
\end{center}
\caption{\label{patternNestedExample}Example of nested structure required for exact completion.  The condition states that there should exist and ordering of the self avoiding walks so that the steps of walk $\omega^{(t+1)}$ are enclosed in the steps of walk $\omega^{(t)}$. }
\end{figure}

\subsubsection{\label{generalCase}General case}

In a first time, we will only consider almost self avoiding walks that consist of sequences of the form $F^{k_1} R F^{k_2} L F^{k_3} R F^{k_4} $ where $F,  R$ and $L$ respectively denote \emph{``forward"}, \emph{``right"} or \emph{``left"} moves.  For simplicity we let all the path start at position (0,0) and we assume that the overlap is minimal. 
In this case we have a staircase structure similar to the one shown in Fig.~\ref{rank3graph01}.  In the general low rank setting, as expected, we will need additional path (all of them spanning the whole range of column and row indices). As a result, we will have to consider the extension from walks to lattice trees and the non overlapping condition that we had in rank-2 will turn into a condition of minimal overlap and of non occlusion when such an overlap is unavoidable.  We start by clarifying those ideas below. 

\begin{condition}[\label{nonOcclusion}Non occlusion]
For a collection of lattice trees $\left\{T^{(0)}, \ldots, T^{(r)}\right\}$ on $\Lambda_{m,n}\subset \mathbb{Z}^2$, we say that the trees are non occluding if when there is an overlap between two trees $T^{(i)}$ and $T^{(j)}$,  the set of row/column indices that is spanned by the intersection of the two trees has to be covered by another subtree from either $T^{(i)}$ or $T^{(j)}$ that does not lie in the intersection. Mathematically,  let $P = \left\{T^{(i)}\cap T^{(j)}\right\}$ and let
\begin{align}
S_1(T^{(i)}\cap T^{(j)}) = \left\{k\in \Lambda_{m,n}\;|\; \exists \omega = (\omega_1, \omega_2)\in T^{(i)}\cap T^{(j)}\; \text{with} \; \omega_1 = k \right\}\end{align} 
\begin{align}
S_2(T^{(i)}\cap T^{(j)}) = \left\{\ell \in \Lambda_{m,n} \;|\; \exists \omega = (\omega_1, \omega_2)\in T^{(i)}\cap T^{(j)}\; \text{with} \; \omega_2 = \ell \right\}\end{align}
For every $k\in  S_1(T^{(i)}\cap T^{(j)})$ $\ell\in  S_2\left(T^{(i)}\cap T^{(j)})\right)$ there must exist $T'\subset T^{(i)}\cup T^{(j)}\setminus (T^{(i)}\cap T^{(j)})$ with $k\in S_1(T')$ and $\ell\in S_2(T')$.
\end{condition}

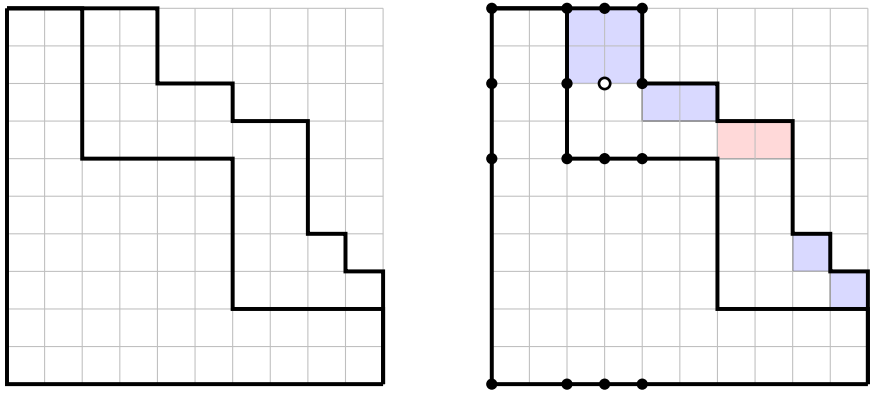
\begin{figure}
\centering 

\begin{minipage}{.45\linewidth}
\centering

\begin{tikzpicture}[scale=0.5]

%
%
%
%
%
%
%

\foreach \x in {0,...,10} \draw[gray!50] (\x,0) -- (\x,10);
\foreach \y in {0,...,10} \draw[gray!50] (0,\y) -- (10,\y);

\draw[line width=1.5pt] 
    (0,10) -- (0,0) -- (10,0) ;

\draw[line width=1.5pt] 
    (0,10) -- (2,10) -- (2, 6) -- (6, 6) -- (6,2)-- (10,2) -- (10,0);

\draw[line width=1.5pt] 
    (0,10) -- (4,10) -- (4, 8) -- (6, 8) -- (6,7) -- (8,7) -- (8,4)-- (9,4)-- (9,3)-- (10,3) -- (10,0);

\end{tikzpicture}
\end{minipage}
\begin{minipage}{.45\linewidth}
\begin{tikzpicture}[scale=0.5]

%
%
%
%
%
%
%

\draw[fill=blue!15] (2,8) rectangle (4,10);

\draw[fill=red!15] (6,6) rectangle (8,7);

\draw[fill=blue!15] (4,7) rectangle (6,8);


\draw[fill=blue!15] (8,3) rectangle (9,4);

\draw[fill=blue!15] (9,2) rectangle (10,3);

\foreach \x in {0,...,10} \draw[gray!50] (\x,0) -- (\x,10);
\foreach \y in {0,...,10} \draw[gray!50] (0,\y) -- (10,\y);

\draw[line width=1.5pt] 
    (0,10) -- (0,0) -- (10,0) ;

\draw[line width=1.5pt] 
    (0,10) -- (2,10) -- (2, 6) -- (6, 6) -- (6,2)-- (10,2) -- (10,0);

\draw[line width=1.5pt] 
    (0,10) -- (4,10) -- (4, 8) -- (6, 8) -- (6,7) -- (8,7) -- (8,4)-- (9,4)-- (9,3)-- (10,3) -- (10,0);
\node at (0,0)[circle,fill,inner sep=1.5pt]{};
\node at (2,0)[circle,fill,inner sep=1.5pt]{};
\node at (3,0)[circle,fill,inner sep=1.5pt]{};
\node at (4,0)[circle,fill,inner sep=1.5pt]{};

\node at (0,10)[circle,fill,inner sep=1.5pt]{};
\node at (2,10)[circle,fill,inner sep=1.5pt]{};
\node at (3,10)[circle,fill,inner sep=1.5pt]{};
\node at (4,10)[circle,fill,inner sep=1.5pt]{};

\node at (0,6)[circle,fill,inner sep=1.5pt]{};
\node at (2,6)[circle,fill,inner sep=1.5pt]{};
\node at (3,6)[circle,fill,inner sep=1.5pt]{};
\node at (4,6)[circle,fill,inner sep=1.5pt]{};

\node at (0,8)[circle,fill,inner sep=1.5pt]{};
\node at (2,8)[circle,fill,inner sep=1.5pt]{};
\node at (3,8)[circle,fill =white, draw=black, inner sep=1.5pt, line width=1pt]{};
\node at (4,8)[circle,fill,inner sep=1.5pt]{};

\end{tikzpicture}
\end{minipage}
\caption{\label{rank3graph01}An example of a $C$-graph for $r=3$ . For the conditions in Proposition~\ref{generalLowRankNestedSteps} to be satisfied we assume that every white cell is of size at least $3$}
\end{figure}

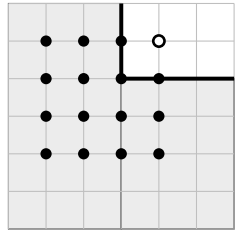
\begin{figure}
\centering 
\begin{minipage}{.3\linewidth}
\centering
\begin{tikzpicture}[scale=0.5]

\draw[fill=gray!15] (0,0) rectangle (3,6);
\draw[fill=gray!15] (3,0) rectangle (6,4);

\foreach \x in {0,...,6} \draw[gray!50] (\x,0) -- (\x,6);
\foreach \y in {0,...,6} \draw[gray!50] (0,\y) -- (6,\y);

\draw[line width=1.5pt] 
   (3,6) -- (3,4)-- (6,4) ;

\node at  (4,5)[circle,fill =white, draw=black, inner sep=1.5pt, line width=1pt]{};
\node at  (4,4)[circle,fill,inner sep=1.5pt]{};
\node at  (3,5)[circle,fill,inner sep=1.5pt]{};
\node at  (3,4)[circle,fill,inner sep=1.5pt]{};
\node at  (2,4)[circle,fill,inner sep=1.5pt]{};
\node at  (1,4)[circle,fill,inner sep=1.5pt]{};

\node at  (1,5)[circle,fill,inner sep=1.5pt]{};
\node at  (2,5)[circle,fill,inner sep=1.5pt]{};

\node at  (1,3)[circle,fill,inner sep=1.5pt]{};
\node at  (2,3)[circle,fill,inner sep=1.5pt]{};
\node at  (3,3)[circle,fill,inner sep=1.5pt]{};
\node at  (4,3)[circle,fill,inner sep=1.5pt]{};

\node at  (1,2)[circle,fill,inner sep=1.5pt]{};
\node at  (2,2)[circle,fill,inner sep=1.5pt]{};
\node at  (3,2)[circle,fill,inner sep=1.5pt]{};
\node at  (4,2)[circle,fill,inner sep=1.5pt]{};

\end{tikzpicture}
\end{minipage}

\caption{Completion of the exterior}
\end{figure}

\begin{figure}

\begin{minipage}{.45\linewidth}
\centering

\begin{tikzpicture}[scale=0.5]

\draw[fill=blue!5] (2,8) rectangle (4,10);

\draw[fill=blue!10] (2,6) rectangle (6,8);

\draw[fill=blue!20] (6,4) rectangle (8,8);

\foreach \x in {0,...,10} \draw[gray!50] (\x,0) -- (\x,10);
\foreach \y in {0,...,10} \draw[gray!50] (0,\y) -- (10,\y);

\draw[line width=1.5pt] 
    (0,10) -- (0,0) -- (10,0) ;

\draw[line width=1.5pt] 
    (0,10) -- (2,10) -- (2, 6) -- (6, 6) -- (6,2)-- (10,2) -- (10,0);

\draw[line width=1.5pt] 
    (0,10) -- (4,10) -- (4, 8) -- (8, 8) -- (8,6) -- (8,4) -- (10,4)-- (10,2);

\node at (0,8)[circle,fill,inner sep=1.5pt]{};
\node at (2,8)[circle,fill,inner sep=1.5pt]{};
\node at (3,8)[circle,fill =white, draw=black, inner sep=1.5pt, line width=1pt]{};
\node at (4,8)[circle,fill,inner sep=1.5pt]{};

\node at (0,10)[circle,fill,inner sep=1.5pt]{};
\node at (2,10)[circle,fill,inner sep=1.5pt]{};
\node at (3,10)[circle,fill,inner sep=1.5pt]{};
\node at (4,10)[circle,fill,inner sep=1.5pt]{};

\node at (0,6)[circle,fill,inner sep=1.5pt]{};
\node at (2,6)[circle,fill,inner sep=1.5pt]{};
\node at (3,6)[circle,fill,inner sep=1.5pt]{};
\node at (4,6)[circle,fill,inner sep=1.5pt]{};

\node at (0,0)[circle,fill,inner sep=1.5pt]{};
\node at (2,0)[circle,fill,inner sep=1.5pt]{};
\node at (3,0)[circle,fill,inner sep=1.5pt]{};
\node at (4,0)[circle,fill,inner sep=1.5pt]{};

\end{tikzpicture}
\end{minipage}
\hspace{.1cm}
\begin{minipage}{.45\linewidth}
\centering

\begin{tikzpicture}[scale=0.5]

\draw[fill=blue!5] (2,8) rectangle (4,10);

\draw[fill=blue!10] (2,6) rectangle (6,8);

\draw[fill=blue!20] (6,4) rectangle (8,8);

\foreach \x in {0,...,10} \draw[gray!50] (\x,0) -- (\x,10);
\foreach \y in {0,...,10} \draw[gray!50] (0,\y) -- (10,\y);

\draw[line width=1.5pt] 
    (0,10) -- (0,0) -- (10,0) ;

\draw[line width=1.5pt] 
    (0,10) -- (2,10) -- (2, 6) -- (6, 6) -- (6,2)-- (10,2) -- (10,0);

\draw[line width=1.5pt] 
    (0,10) -- (4,10) -- (4, 8) -- (8, 8) -- (8,6) -- (8,4) -- (10,4)-- (10,2);

\node at (0,8)[circle,fill,inner sep=1.5pt]{};
\node at (2,8)[circle,fill,inner sep=1.5pt]{};
\node at (3,8)[circle,fill,inner sep=1.5pt]{};
\node at (8,8)[circle,fill,inner sep=1.5pt]{};

\node at (0,7)[circle,fill,inner sep=1.5pt]{};
\node at (2,7)[circle,fill,inner sep=1.5pt]{};
\node at (3,7)[circle,fill =white, draw=black, inner sep=1.5pt, line width=1pt]{};
\node at (8,7)[circle,fill,inner sep=1.5pt]{};

\node at (0,6)[circle,fill,inner sep=1.5pt]{};
\node at (2,6)[circle,fill,inner sep=1.5pt]{};
\node at (3,6)[circle,fill,inner sep=1.5pt]{};
\node at (8,6)[circle,fill,inner sep=1.5pt]{};

\node at (0,0)[circle,fill,inner sep=1.5pt]{};
\node at (2,0)[circle,fill,inner sep=1.5pt]{};
\node at (3,0)[circle,fill,inner sep=1.5pt]{};
\node at (8,0)[circle,fill,inner sep=1.5pt]{};

\end{tikzpicture}
\end{minipage}
\caption{\label{rank302a}The completion of the inner path region after removing the nested steps.  We assume that every (white) cell includes at least $3$ entries from the adjacency matrix. }
\end{figure}
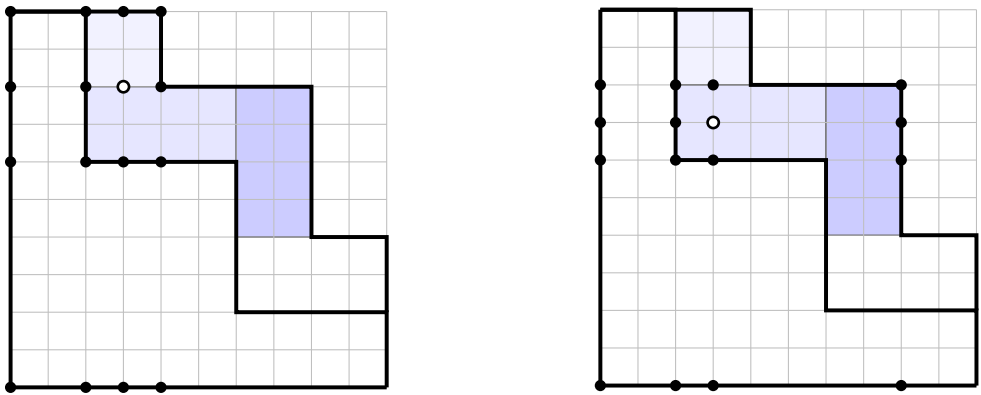

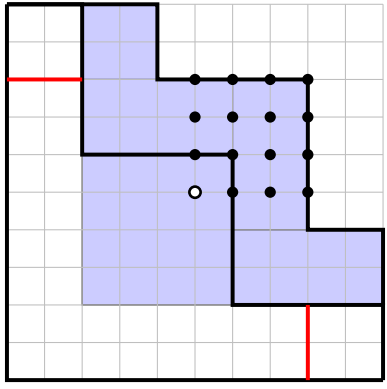
\begin{figure}
\centering 

\begin{minipage}{.33\linewidth}
\centering

\begin{tikzpicture}[scale=0.5]

\draw[fill=blue!20] (2,8) rectangle (4,10);

\draw[fill=blue!20] (2,6) rectangle (6,8);

\draw[fill=blue!20] (6,4) rectangle (8,8);
\draw[fill=blue!20] (2,2) rectangle (6,6);

\draw[fill=blue!20] (6,2) rectangle (10,4);

\foreach \x in {0,...,10} \draw[gray!50] (\x,0) -- (\x,10);
\foreach \y in {0,...,10} \draw[gray!50] (0,\y) -- (10,\y);

\draw[line width=1.5pt] 
    (0,10) -- (0,0) -- (10,0) ;

\draw[line width=1.5pt] 
    (0,10) -- (2,10) -- (2, 6) -- (6, 6) -- (6,2)-- (10,2) -- (10,0);

\draw[line width=1.5pt] 
    (0,10) -- (4,10) -- (4, 8) -- (8, 8) -- (8,6) -- (8,4) -- (10,4)-- (10,2);

\node at (5,5)[circle,fill =white, draw=black, inner sep=1.5pt, line width=1pt]{};
\node at (5,6)[circle,fill,inner sep=1.5pt]{};
\node at (6,5)[circle,fill,inner sep=1.5pt]{};
\node at (6,6)[circle,fill,inner sep=1.5pt]{};
\node at (6,7)[circle,fill,inner sep=1.5pt]{};
\node at (7,6)[circle,fill,inner sep=1.5pt]{};
\node at (7,7)[circle,fill,inner sep=1.5pt]{};
\node at (5,7)[circle,fill,inner sep=1.5pt]{};
\node at (5,8)[circle,fill,inner sep=1.5pt]{};
\node at (7,5)[circle,fill,inner sep=1.5pt]{};
\node at (8,5)[circle,fill,inner sep=1.5pt]{};
\node at (8,6)[circle,fill,inner sep=1.5pt]{};
\node at (8,7)[circle,fill,inner sep=1.5pt]{};

\node at (6,8)[circle,fill,inner sep=1.5pt]{};
\node at (7,8)[circle,fill,inner sep=1.5pt]{};
\node at (8,8)[circle,fill,inner sep=1.5pt]{};

\draw[line width=1.5pt, red] 
    (0,8) -- (2,8) ;

\draw[line width=1.5pt, red] 
    (8,0) -- (8,2) ;

\end{tikzpicture}
\end{minipage}

\caption{\label{rank302b}The completion of the inner path region after removing the nested steps.  We assume that every (white) cell includes at least $3$ entries from the adjacency matrix.  The red segments correspond to the walks $\tilde{\omega}^{(j)}$ that are added to satisfy the non occlusion condition (see Condition~\ref{nonOcclusion} as well as Proposition~\ref{generalLowRankNestedSteps})}
\end{figure}

\begin{condition}[\label{nestedStep}Nested steps]

Let $\omega, \omega'$ denote two self avoiding paths that span the whole set of row and column indices in $\Lambda_{m,n}$.  Further let $\omega(j)$ denote any position in the path such that either $\omega_2(j) = m$ or $\omega(j) - \omega(j-1) = e_1$ and $\omega(j+1)-\omega(j) = e_2$.  Let us use $j'$ such that $j' = \max\left\{k| \langle \omega - \omega(j),e_1\rangle =  0\; \text{for all $\omega \in [\omega(j), \omega(k)]$}\right\}$ and let us label $j''$ the index $j'' = \max\left\{\ell| \langle \omega(j') - \omega,e_2\rangle \neq 0\; \text{for all $\omega \in [\omega(j'), \omega(\ell)]$}\right\}$. Then the sequence $(j,j',j'')$ must be such that for all $\ell'$ 
\begin{align}
x_{j}' = \max \left\{x \; \text{\upshape such that }[x, \omega^{(\ell)}_2(j)]\in \omega^{(\ell')}\right\}\\
y_{j}' = \max \left\{y\; \text{\upshape such that }[\omega^{(\ell)}_1(j''), y] \in \omega^{(\ell')}\right\}
\end{align}
\begin{align}
\underline{y} = \min \left\{y| (x_j' y)\in \omega^{(\ell')}\text{for some $x\in \mathbb{Z}$}\right\}\\
\underline{x} = \min \left\{x| (x, y_j')\in \omega^{(\ell')}\text{for some $y\in \mathbb{Z}$}\right\}
\end{align}

We say that the paths $\omega^{(\ell)}$ and $\omega^{(\ell')}$ satisfy a nested step structure if  $\underline{y}\leq y$ and $\underline{x}\leq x$. 
\end{condition}

We say that a set of paths $\left\{\omega^{(\ell)}\right\}$ satisfy a nested step structure if there exist an ordering of the paths $\omega^{(0)}, \ldots, \omega^{(r)}$ such that every pair $\omega^{(\ell')}, \omega^{(\ell)}$ for every $\ell'<\ell$ satisfy the nested step condition~\ref{nestedStep}.

Condition~\ref{nestedStep} garantees the existence of a sequence of nested steps in the adjacency matrix which in turn makes it possible to iteratively complete the entries as shown in the proof of Proposition~\ref{generalLowRankNestedSteps} 

\begin{condition}[\label{connectionCondition}Connection condition]
For every sequence $(j,j',j'')$ of left-right-left turns in $\omega^{(\ell)}$ that does not satisfy the nested step condition~\ref{nestedStep},  we require that there exist a sequence of turns $(k_1,k_2,k_3,k_4)$ in $\omega^{(\ell')}$ such that $\omega_2^{(\ell-1)}(k_1) = \omega_2^{(\ell)}(j)$, $\omega_1^{(\ell-1)}(k_3) = \omega_1^{(\ell)}(j)$,   $\omega_1^{(\ell-1)}(k_4) = \omega_1^{(\ell)}(j)$ and $\omega_1^{(\ell-1)}(k_5) = \omega_1^{(\ell)}(j'')$ as well as 
\begin{itemize}
\item for any $\omega\in [\omega^{(\ell-1)}(k_1), \omega^{(\ell-1)}(k_2)]$, we have $\langle \omega-\omega^{(\ell-1)}(k_1), e_1\rangle  = 0$
\item for any $\omega\in [\omega^{(\ell-1)}(k_2), \omega^{(\ell-1)}(k_3)]$, we have $\langle \omega-\omega^{(\ell-1)}(k_2), e_2\rangle =0$
\item for any $\omega\in [\omega^{(\ell-1)}(k_3), \omega^{(\ell-1)}(k_4)]$, we have $\langle \omega-\omega^{(\ell-1)}(k_3), e_1\rangle =0$
\item for any $\omega\in [\omega^{(\ell-1)}(k_4), \omega^{(\ell-1)}(k_5)]$, we have $\langle \omega-\omega^{(\ell-1)}(k_4), e_2\rangle =0$
\end{itemize}
\end{condition}

The connection condition ensures that when a step from walk $\omega^{(\ell)}$ is not nested entirely nested in a single step of the walk $\omega^{(\ell-1)}$, this step never spans more than two steps from $\omega^{(\ell-1)}$ as shown in Fig.~\ref{rank3graph01}. 

\begin{definition}[$C$-graph]
A graph $G$ is a $C$-graph for some $r>1$ if the subgraph $\mathcal{G}$ of the lattice graph generated from $\text{\upshape supp}(A_G)$ where $A_G$ is the bi-adjacency matrix of $G$ is of the form $\mathcal{G} = \cup_{i=1}^r \omega^{(r)}\cup_{j=1}^{s}\tilde{\omega}^{(j)}$ where $\left\{\omega^{(\ell)}\right\}_{\ell=1}^r$ are self avoiding walks that obey the nested steps and connection conditions~\ref{nestedStep}, ~\ref{connectionCondition} and $\tilde{\omega}^{s}$ are walks that are added to satisfy the non occlusion condition~\ref{nonOcclusion}.  Moreover we require 
\begin{align}
\min_{t, t'}\left\{|\omega(t)- \omega'(t')| \right\} > r
\end{align}
\end{definition}

We are now ready to state the extension of Proposition~\ref{propositionCycleSelfAvoiding01} to the general low rank completion setting. 

\begin{proposition}\label{generalLowRankNestedSteps}
Let $r$ be fixed and let $G = ([m], [n], E)$ be any bipartite graph that is isomorphic to a $C(r)$-graph $G^*$. Then any rank-$r$ matrix $ X_0$ can be recovered uniquely and stably from the knowledge of $(X_0)_{G}$ through the level $O(\kappa)$ of the sos hierarchy where $\kappa$ is the number of steps in the subgraph of the lattice graph defined from the support of the bi-adjacency matrix of $G^*$. 
\end{proposition}

\begin{proof}
The difficulty is to complete the interior.  Once the interior has been completed, since there are at least $r$ paths,  since the paths are distant by at least $3$ cells each and since the polyomino spans the whole set of row and columns indices, the exterior can always be completed as shown in Fig.~\ref{rank3graph01}

For each step of the trajectory, either the step spans a vertical wall and a horizontal wall that are intersecting (blue squares in Fig.~\ref{rank3graph01}) or the step spans a vertical wall and a horizontal wall that are not intersecting (red square in Fig.~\ref{rank3graph01}). We discuss each of the situations below.   

\begin{itemize}
\item For any blue square, relying on the nested step condition~\ref{nestedStep}, we can complete the interior of the square as follows. We let $(j,j',j'')$ the indices associated to the limit of the step.  For any $(m,n)\in \mathbb{Z}^2$ with $\omega^{(\ell)}_1(j)<m<\omega^{(\ell)}_1(j')$, $\omega_2(j'')<n<\omega_2(j')$ by assumption we always have $K^{(\ell-r+1)}, \ldots K^{(\ell-1)}$ as well as $M^{(\ell-r+1)}, \ldots M^{(\ell-1)}$ such that $(K_j, n)\cup (m, M_\ell)\cup (K_j, M_\ell)$ belong to the subgraph $\mathcal{G}$ of the lattice graph   For each of those step, following the proof of Proposition~\ref{propositionCycleSelfAvoiding01}, one can express the monomial $X_{ij} - (X_0)_{ij}$ from the truncated ideal $I_M^{L}$ generated by the minors and the completion constraints.  

\item After removing all the nested steps, we are left with two parallel staircases such as shown in Fig.~\ref{rank302a} which can be completed by relying on the pattern (either vertical or horizontal) shown in that same figure. 
\end{itemize}

After the subregion located between the first two staircases (i.e. the ones with the smallest steps) has been completed, one can proceed recursively with the larger staircases, by filling the rectangles located below each step (relying on the region completed at the previous step) until we are again left with parallel staircases.  If the subregion corresponds to indices (either row or column) that were shared by several walks, we make use of the additional subpaths $\tilde{\omega}$ (see Fig.~\ref{rank302b}). 

\end{proof}

\section{Connections with existing work}

\subsection{\label{pseudoRandomSchemes}(Pseudo-)random sampling and incoherence }

The combination of random sampling and random or sufficiently incoherent singular vectors makes it possible to rely on classical convex optimization tools such as the nuclear norm. or semidefinite programming as well as spectral (projection) techniques. In~\cite{candes2008exact} the authors guarantee the recovery of the unknown matrix through nuclear norm minimization, as soon as $m\geq C n^{6/5}r\log n$ for random singular vectors and $m\geq C\mu rn\log(n)$ where $\mu$ depends on the incoherence of the singular vectors.  The result was later improved in~\cite{candes2010power} and~\cite{recht2011simpler}.  In both of those results, the row and column subspaces are requried to be sufficiently incoherent. Namely if we let $U^k$ (resp. $V^k$) to denote the $k^{th}$ row of $U$ (resp. $V$) in column format, those results require 
\begin{align}
\|U^k\|^2 \leq \frac{\mu_0 r}{m}, \; \forall k, \quad \|V^k\|^2 \leq \frac{\mu_0 r}{n}, \; \forall k.\label{incoherence01}
\end{align} 
as well as 
\begin{align}
\begin{split}
\left|\langle \bs e_a,  UU^\intercal, \bs e_b\rangle - \frac{r}{m} \mathbf{1}_{a=b} \right|\leq \mu_1 \frac{\sqrt{r}}{m}, \quad \forall 1\leq a, b\leq m.\\
\left|\langle \bs e_a,  VV^\intercal, \bs e_b\rangle - \frac{r}{n} \mathbf{1}_{a=b} \right|\leq \mu_1 \frac{\sqrt{r}}{n}, \quad \forall 1\leq a, b\leq m.
\end{split}\label{incoherence02}
\end{align}
In~\cite{keshavan2009matrix, keshavan2010matrix} the authors suggest to start from the estimate $(X_0)_{\Omega}$ and then apply the following steps: (i) trim the estimate by setting to zero all the columns that contain more than $2|\Omega|/n$ entries and all the rows that contain more than $2|\Omega|/m$ entries (ii) project the resulting estimate on the manifold of rank $r$ matrices by only retaining the largest $r$ singular values and (iii) clean the residual error by minimizing the discrepancy
\begin{align}
\min_{S\in \mathbb{R}^{r\times r}} \frac{1}{2}\sum_{(i,j)\in \Omega} \left((X_0)_{ij} - (XSY^\intercal)_{ij}\right)^2
\end{align}
where $XSY^\intercal$ is the result of the projection on the rank-$r$ manifold. The main result of the paper requires the set $\Omega$ to be sampled uniformly at random and the convergence of the scheme holds with high probability under incoherence conditions equivalent to the ones in~\eqref{incoherence01} and~\eqref{incoherence02} and as soon as $|\Omega|\geq nr (\sigma_{\max}/\sigma_{\min})^2 \mu'$ where $\mu'$ depends on the (in)coherences of the singular vectors. The result is extended to the noisy setting in~\cite{keshavan2010matrix} where a stability estimate is derived. 

A natural extension of the results discussed above consists in studying the completion of matrices defined on deterministic graphs that closely mimics randomness.  A family of such graphs is the family of expander graphs which are sparse despite having strong connectivity and mixing properties.  Expander graphs guarantee random like properties that make them useful in compressed sensing.  In~\cite{bhojanapalli2014universal}, the authors replace the uniform sampling assumption by a sampling scheme generated from the edge distribution of sufficiently strong expanders known as Ramanujan graphs.  Ramanujan graphs are associated with the largest possible spectral gap among expanders (i.e.  $d$-regular bipartite graph with $\sigma_2(G)\leq \sqrt{d-1}$). This can in turn be used with the nuclear norm to provide stable recovery guarantees under an incoherence assumption on the singular vectors of the unknown matrix.  In~\cite{bhojanapalli2014universal} the authors certify stable recovery of the matrix $X_0$ as soon as $d\geq C \mu_0^2 r^2$ and $|\Omega|\geq C\mu_0^2 r^2\max \left\{m,n\right\}$.  Following this first connection between expander graphs and matrix completion, a number of corrections and improvements were made including~\cite{burnwal2020deterministic} where the authors reduce the minimum number of measurements per row/column.  The result still requires singular vectors to be sufficiently incoherent (i.e the scalings still involving the incoherence $\mu_0$ defined in~\eqref{incoherence01}).  

In~\cite{heiman2014deterministic}, the authors suggest to recover the matrix through the minimization of the $\gamma_2$-norm defined as 
\begin{align}
\gamma_2(X) = \min_{UV^\intercal  = X} \|U\|_{\ell_2\rightarrow \ell_{\infty}^m} \|V\|_{\ell_1^n \rightarrow \ell_2}
\end{align}
The norm can be computed in polynomial time using convex programming.  When the observation are generated from the edges of a $d$-regular graph with second eigenvalue bounded by $\lambda$, for a square $n\times n$ matrix, the following estimate can be derived
\begin{align}
\frac{1}{n^2}\sum_{i,j} ((X_0)_{ij} - \hat{X}_{ij})^2 \leq c\gamma_2^2(X_0)\frac{\lambda}{d}\label{recoveryEstimateBasedGamma2norm}
\end{align}
The result does not require the usual incoherence assumptions on the matrix but the dependence on the $\gamma_2$-norm that appears in~\eqref{recoveryEstimateBasedGamma2norm} is somewhat equivalent in the sense that this estimate depends on the magnitude of the entries (see the discussion in section~\ref{discussionAlgebraicMC}) without providing intuition on the sampling patterns or graphs for which exact completion is possible.  Moreover, this type of recovery estimate does not make it possible to certify exact recovery except in the asymptotic regime when $\lambda = o(d)$ (e.g. Ramanujan) and for $\gamma_2(X)$ that does not grow with $d$.  The authors extend their result by providing an error estimate in the case of a general probability distribution $P$ on the entries. In this case the derive a bound of the form
\begin{align}
\sum_{(i,j)} p_{ij} ((X_0)_{ij} - \hat{X}_{ij})^2 \leq c\gamma_2(X_0)/\sqrt{d}.
\end{align}
Finally in~\cite{chen2014coherent} the authors show that nuclear norm minimization can be used to complete arbitrary rank $r$ matrices from $O(rn \log^2(n))$ entries provided that the sampling is proportionnal to local row and column coherences.  None of the above results really shed light on the characterization of sampling schemes for which exact completion can be guaranteed.  By being unable to certify exact completion outside asymptotic regimes, they implicitely seem to suggest that sampling schemes based on regular graphs or generated uniformly at random might not be the keys to understand and characterize completable sampling patterns. 

\subsection{\label{discussionAlgebraicMC}Deterministic/algebraic constructions}

As explained above, despite clear theoretical interests (in particular when combined with convex optimization) the incoherence assumptions and/or $\gamma_2$-norm do not provide intuition on what ``drives" the completability of low rank matrices.  I.e. what makes the difference between the patterns that can be completed and those that cannot,  and among the patterns that can be completed in theory, the difference between those that can be completed efficiently and those for which the exact completion, while possible in theory, remains out of reach in practice.  
Around Theorem~\ref{theoremOnlystructure01}, a number of approaches have attempted to characterize successful sampling patterns per se,  that is to say with respect to the structure itself and without the need for additional assumptions on the unknown matrix. Those approaches can roughly be organized into two groups: the first group (which includes~\cite{kiraly2015algebraic}) relies on combinatorics and matroid theory, notions from algebraic geometry such as the characterization of the determinantal variety and the Jacobian of the map $\mathcal{A}\;:\; U,V\mapsto A = UV^\intercal$ and/or the Grassmanian. The second relies on rigidity theory and the \emph{distance geometry problem} (a.k.a. \emph{graph realization problem}). 

Both approaches usually make the distinction between \emph{finite completability} (or \emph{local completability} in rigidity theory) and \emph{unique completability} (or \emph{global completability} in rigidity theory).  The terms \emph{finite} or \emph{unique completability} are used to denote settings in which completing the matrix uniquely given the pattern is not possible but the number of completion remains finite.  Conversely the authors use the terms \emph{global} or \emph{unique completability} when the low rank completion is guaranteed to be unique.  The classical references along those two directions can be summarized as follows. 

For finite completion, the authors in~\cite{pimentel2016characterization} require $r+1$ entries to be observed in each column and they require the existence of a $m\times r(m-r)$ submatrix $\tilde{M}$ of the mask $M$ such that every matrix formed by a subset of columns of $\tilde{M}$ has a number of non zero rows that is larger than $n/r + r$.  Similarly, for unique completion,  they require the existence of a partition of $M$ into two disjoint submatrices $\tilde{M}_1$ and $\tilde{M}_2$ of respective sizes $m\times r(m-r)$ and $m\times (m-r)$ such that the first mask $\tilde{M}_1$ satisfies the finite completability condition and for every submatrix formed by a subset of the columns of $\tilde{M}_2$ has a number of non zero rows that is at least $n+r$.   

A classical series of papers is~\cite{kiraly2012combinatorial, kiraly2013error, kiraly2015algebraic}.  The closest to this work in~\cite{kiraly2015algebraic} which derives several general finite and unique completability results based on matroid theory and the analysis of the Jacobian map $\mathcal{A}\;:\; U,V\mapsto A = UV^\intercal$.  We recall some of the most useful ideas in section~\ref{mathematicalPreliminaries}.  Among the main contributions of the paper, one must mention Theorems 10 and 17 which formalize the idea that completability solely depend on the structure of the measurements and not on the unknown matrix itself.  Additional key contributions of the paper include Algorithms 1 and 2 which compute the notions of rank $r$ finitely completable closure (Algorithm 1) and generic stress rank (Algorithm 2) that are key in certifying finite or unique completability.  Despite the value of the tools derived in the paper,  the characterization of completable graph remains relatively limited and the completability conditions relatively general.  The most relevant to this paper can be derived from Proposition 36 and Corollary 39 and says that if a pattern $\Omega$ is finitely completable in rank $r$, it is necessarily $r$-edge connected and has degree at least $r$.  This is in accordance with the patterns that are derived in this paper.

In the context of rigidity theory based approaches,  the authors usually start by connecting the graph rigidity/distance geometry problems to the completion of Gram (positive semidefinite) matrices.  Such an approach is followed in~\cite{singer2010uniqueness} as well as~\cite{hendrickson1992conditions} and~\cite{jackson2014combinatorial}.  The results from this second line of work usually rely on the notion of infinitesimal motion. I.e when considering samples from the distance matrix $d_{ij} = \|p_i - p_j\|$,  we can consider the smooth motions $p(t) = [p_1(t), \ldots, p_n(t)]$ that preserve the pairwise distances, i.e. such that $\frac{d}{dt}\|p_i - p_j\|^2 = 0$ for all $(i,j)\in E$.  Such a definition thus implies the well known relations 
\begin{align}
\langle p_i - p_j,  \dot{p}_i - \dot{p}_j \rangle  = 0, \quad \text{for all $(i,j)\in E$.}\label{systemInfinitesimalMotions}
\end{align}
This system of equations can then be written as $R_G(p)\dot{p} = 0$, where $R_G(p)$ is known as the rigidity matrix and is of size $m\times dn$.  One can then study the configuration for which no infinitesimal deformation is possible. It is important to note that for every skew-symmetric $d\times d$ matrix $A$ (with $A^\intercal  = -A$),  and for every $b\in \mathbb{R}^d$, we have $\dot{p}_i = Ap_i + b$ is an infiniteseimal motion that satisfies the system~\eqref{systemInfinitesimalMotions}. In the setting of graph rigidity, $A$ represents an orthogonal transformation and $b$ is the translation vector.  Combining the $d$ degrees of freedom of the translation to the $d(d-1)/2$ degrees of freedom of the orthogonal transformation, a configuration is called infinitesimally rigid if $\text{\upshape dim null}(R_G(p)) = d(d+1)/2$.  The authors of~\cite{singer2010uniqueness} use the similarity between the graph rigidity problem and the completion problem to derive a number of general local and global completion results, the stronger ones being derived for the rank-1 completion of Gram matrices.  The most related to this paper are probably Propositions 5.1. and 5.2. which imply that every locally completable matrix contains a graph that spans the whole set of vertices and that is $(d, \frac{d-1}{2})$ sparse which means that every subset of $n'\leq n$ spans at most $dn' - \frac{d-1}{2}$ edges.  The paper also provides algorihtms to test local and global completion.  

A number of other papers address this connection between the graph rigidity problem and matrix completion, starting with the completion of positive semidefinite Gram matrices. One of the most relevant to us are~\cite{hendrickson1992conditions} and~\cite{jackson2014combinatorial} in which a number of guarantees are derived.  The most related to this paper include~\cite{hendrickson1992conditions}, Theorem 3.1., where the authors show that if a configuration is globally rigid for a generic configuration, then the graph must be (d+1)-connected (which means that it remains connected after removing $d+1$ vertices) (also see Theorem 1.2.1 in~\cite{cucuringu2012graph} and the discussions that follow for details).  As indicated in~\cite{singer2010uniqueness}, it is important to keep in mind that the global rigidity problem is different from the uniqueness in matrix completion (the problems are in fact known to be equivalent only at the local level and  when the diagonal entries of the Gram matrix are all known, or equivalently, when all the vertices in the graph contain a loop).  The results can nonetheless be used to develop some intuition as we try to do below. Another interesting line of result can be found in~\cite{jackson2014combinatorial}.  Among the most meaningful results of this paper, one should mention Theorems 13 and 14 (see also Theorem 12 for the weaker notion of local completability). The latter of those statements in particular show that the cluster graph induced by a general multigraph $G = (V, E)$ with $|V|$ and $|E|$ larger than $2$ is globally completable in $\mathbb{R}^d$ if and only if the graph $H$ is highly ${d\choose 2}$ tree-connected.  A graph is called ${d\choose 2}$-tree connected if for all partitions $\mathcal{P} = \left\{X_1, X_2, \ldots , X_t\right\}$ of $V$,  if $e_{H}(\mathcal{P})$ denotes the number of edges of $H$ connecting distinct members of $\mathcal{P}$,  we have $e_H(\mathcal{P})\geq {d\choose 2}(t-1)$. More interestingly, Results by Nash-Williams and Tutte~\cite{nash1961edge, tutte1961problem} show equivalence between the notion $m$-tree connectivity and the existence of $m$ disjoint spanning trees in the graph $H$.  Although the catch here is that the ${d\choose 2}$-tree connectivity condition is enforced on $H$ while the completability is guaranteed on the cluster graph which is obtained from $H$ as follows: One first replaces each vertex $v\in V$ by the (complete) graph $K^\circ_{d(v)}$ on $d(v)$ vertices to which a loop is added on every vertex.  One then repalce each edge $(s,t)\in E$ by an edge between the clusters $C_s$ and $C_t$ in such a way that no two edges share a common vertex (i.e. the edges are pairwise disjoint).  The catch here is that the (high) ${d\choose 2}$-tree connectivity of $H$ is needed to guarantee the global completability of $G^\circ_H$ and not $H$ itself. Nonetheless, the connection between unique completion and the existence of a sufficient number of non overlapping trees is an interesting hint and seems to be corroborated by the results presented in this paper.

\bibliographystyle{ieeetr}
\bibliography{sample}

@incollection{laurent2008sums,
	author = {Laurent, Monique},
	booktitle = {Emerging applications of algebraic geometry},
	pages = {157--270},
	publisher = {Springer},
	title = {Sums of squares, moment matrices and optimization over polynomials},
	year = {2008}}

@article{biswas2006semidefinite,
	author = {Biswas, Pratik and Lian, Tzu-Chen and Wang, Ta-Chung and Ye, Yinyu},
	journal = {ACM Transactions on Sensor Networks (TOSN)},
	number = {2},
	pages = {188--220},
	publisher = {ACM New York, NY, USA},
	title = {Semidefinite programming based algorithms for sensor network localization},
	volume = {2},
	year = {2006}}

@article{singer2008remark,
	author = {Singer, Amit},
	journal = {Proceedings of the National Academy of Sciences},
	number = {28},
	pages = {9507--9511},
	publisher = {National Academy of Sciences},
	title = {A remark on global positioning from local distances},
	volume = {105},
	year = {2008}}

@article{hadani2011representation,
	author = {Hadani, Ronny and Singer, Amit},
	journal = {Annals of mathematics},
	number = {2},
	pages = {1219},
	title = {Representation theoretic patterns in three dimensional Cryo-Electron Microscopy I: The intrinsic reconstitution algorithm},
	volume = {174},
	year = {2011}}

@article{chen2004recovering,
	author = {Chen, Pei and Suter, David},
	journal = {IEEE transactions on pattern analysis and machine intelligence},
	number = {8},
	pages = {1051--1063},
	publisher = {IEEE},
	title = {Recovering the missing components in a large noisy low-rank matrix: Application to SFM},
	volume = {26},
	year = {2004}}

@article{goldberg1992using,
	author = {Goldberg, David and Nichols, David and Oki, Brian M and Terry, Douglas},
	journal = {Communications of the ACM},
	number = {12},
	pages = {61--70},
	publisher = {ACM New York, NY, USA},
	title = {Using collaborative filtering to weave an information tapestry},
	volume = {35},
	year = {1992}}

@article{cai2016structured,
	author = {Cai, Tianxi and Cai, T Tony and Zhang, Anru},
	journal = {Journal of the American Statistical Association},
	number = {514},
	pages = {621--633},
	publisher = {Taylor \& Francis},
	title = {Structured matrix completion with applications to genomic data integration},
	volume = {111},
	year = {2016}}

@article{tutte1961problem,
	author = {Tutte, William Thomas},
	journal = {Journal of the London Mathematical Society},
	number = {1},
	pages = {221--230},
	publisher = {Oxford University Press},
	title = {On the problem of decomposing a graph into n connected factors},
	volume = {1},
	year = {1961}}

@article{nash1961edge,
	author = {Nash-Williams, C St JA},
	journal = {Journal of the London Mathematical Society},
	number = {1},
	pages = {445--450},
	publisher = {Oxford University Press},
	title = {Edge-disjoint spanning trees of finite graphs},
	volume = {1},
	year = {1961}}

@article{kiraly2013error,
	author = {Kiraly, Franz and Theran, Louis},
	journal = {Advances in Neural Information Processing Systems},
	title = {Error-minimizing estimates and universal entry-wise error bounds for low-rank matrix completion},
	volume = {26},
	year = {2013}}

@article{kiraly2012combinatorial,
	author = {Kir{\'a}ly, Franz and Tomioka, Ryota},
	journal = {arXiv preprint arXiv:1206.6470},
	title = {A combinatorial algebraic approach for the identifiability of low-rank matrix completion},
	year = {2012}}

@inproceedings{chen2014coherent,
	author = {Chen, Yudong and Bhojanapalli, Srinadh and Sanghavi, Sujay and Ward, Rachel},
	booktitle = {International Conference on Machine Learning},
	organization = {PMLR},
	pages = {674--682},
	title = {Coherent matrix completion},
	year = {2014}}

@article{heiman2014deterministic,
	author = {Heiman, Eyal and Schechtman, Gideon and Shraibman, Adi},
	journal = {Random Structures \& Algorithms},
	number = {2},
	pages = {306--317},
	publisher = {Wiley Online Library},
	title = {Deterministic algorithms for matrix completion},
	volume = {45},
	year = {2014}}

@inproceedings{burnwal2020exact,
	author = {Burnwal, Shantanu Prasad and Vidyasagar, Mathukumalli},
	booktitle = {2020 American Control Conference (ACC)},
	organization = {IEEE},
	pages = {2203--2206},
	title = {Exact Completion of Rectangular Matrices Using Ramanujan Bigraphs},
	year = {2020}}

@inproceedings{butnwal2019some,
	author = {Butnwal, Shantanu Prasad and Vidyasagar, Mathukumalli},
	booktitle = {2019 Sixth Indian Control Conference (ICC)},
	organization = {IEEE},
	pages = {403--406},
	title = {Some Observations about Ramanujan Graphs With Applications to Matrix Completion},
	year = {2019}}

@inproceedings{burnwal2019construction,
	author = {Burnwal, Shantanu Prasad and Vidyasagar, Mathukumalli},
	booktitle = {2019 American Control Conference (ACC)},
	organization = {IEEE},
	pages = {4814--4817},
	title = {Construction of High-Degree Ramanujan Graphs With Applications to Matrix Completion},
	year = {2019}}

@article{burnwal2020deterministic,
	author = {Burnwal, Shantanu Prasad and Vidyasagar, Mathukumalli},
	journal = {IEEE Transactions on Signal Processing},
	pages = {3834--3848},
	publisher = {IEEE},
	title = {Deterministic completion of rectangular matrices using asymmetric Ramanujan graphs: Exact and stable recovery},
	volume = {68},
	year = {2020}}

@inproceedings{bhojanapalli2014universal,
	author = {Bhojanapalli, Srinadh and Jain, Prateek},
	booktitle = {International Conference on Machine Learning},
	organization = {PMLR},
	pages = {1881--1889},
	title = {Universal matrix completion},
	year = {2014}}

@article{keshavan2010matrix,
	author = {Keshavan, Raghunandan H and Montanari, Andrea and Oh, Sewoong},
	journal = {IEEE transactions on information theory},
	number = {6},
	pages = {2980--2998},
	publisher = {IEEE},
	title = {Matrix completion from a few entries},
	volume = {56},
	year = {2010}}

@article{keshavan2009matrix,
	author = {Keshavan, Raghunandan and Montanari, Andrea and Oh, Sewoong},
	journal = {Advances in neural information processing systems},
	title = {Matrix completion from noisy entries},
	volume = {22},
	year = {2009}}

@article{recht2011simpler,
	author = {Recht, Benjamin},
	journal = {Journal of Machine Learning Research},
	number = {12},
	title = {A simpler approach to matrix completion.},
	volume = {12},
	year = {2011}}

@inproceedings{candes2008exact,
	author = {Candes, Emmanuel J and Recht, Benjamin},
	booktitle = {2008 46th Annual Allerton Conference on Communication, Control, and Computing},
	organization = {IEEE},
	pages = {806--812},
	title = {Exact low-rank matrix completion via convex optimization},
	year = {2008}}

@article{candes2010power,
	author = {Cand{\`e}s, Emmanuel J and Tao, Terence},
	journal = {IEEE transactions on information theory},
	number = {5},
	pages = {2053--2080},
	publisher = {IEEE},
	title = {The power of convex relaxation: Near-optimal matrix completion},
	volume = {56},
	year = {2010}}

@phdthesis{cucuringu2012graph,
	author = {Cucuringu, Mihai},
	school = {Princeton University},
	title = {Graph realization and low-rank matrix completion},
	year = {2012}}

@article{hendrickson1992conditions,
	author = {Hendrickson, Bruce},
	journal = {SIAM journal on computing},
	number = {1},
	pages = {65--84},
	publisher = {SIAM},
	title = {Conditions for unique graph realizations},
	volume = {21},
	year = {1992}}

@article{jackson2014combinatorial,
	author = {Jackson, Bill and Jord{\'a}n, Tibor and Tanigawa, Shin-ichi},
	journal = {SIAM Journal on Discrete Mathematics},
	number = {4},
	pages = {1797--1819},
	publisher = {SIAM},
	title = {Combinatorial conditions for the unique completability of low-rank matrices},
	volume = {28},
	year = {2014}}

@article{singer2010uniqueness,
	author = {Singer, Amit and Cucuringu, Mihai},
	journal = {SIAM Journal on Matrix Analysis and Applications},
	number = {4},
	pages = {1621--1641},
	publisher = {SIAM},
	title = {Uniqueness of low-rank matrix completion by rigidity theory},
	volume = {31},
	year = {2010}}

@article{pimentel2016characterization,
	author = {Pimentel-Alarc{\'o}n, Daniel L and Boston, Nigel and Nowak, Robert D},
	journal = {IEEE Journal of Selected Topics in Signal Processing},
	number = {4},
	pages = {623--636},
	publisher = {IEEE},
	title = {A characterization of deterministic sampling patterns for low-rank matrix completion},
	volume = {10},
	year = {2016}}

@book{van2015statistical,
	author = {Van Rensburg, EJ Janse},
	publisher = {Oxford University Press},
	title = {The statistical mechanics of interacting walks, polygons, animals and vesicles},
	year = {2015}}

@article{cosse2021stable,
	author = {Cosse, Augustin and Demanet, Laurent},
	journal = {Foundations of Computational Mathematics},
	number = {4},
	pages = {891--940},
	publisher = {Springer},
	title = {Stable rank-one matrix completion is solved by the level 2 Lasserre relaxation},
	volume = {21},
	year = {2021}}

@article{tullekenpolyominoes2,
	author = {TULLEKEN, HERMAN},
	title = {POLYOMINOES2. 2}}

@book{polypoly2009,
	edition = {1},
	editor = {Anthony J. Guttman},
	publisher = {Springer Dordrecht},
	series = {Lecture Notes in Physics},
	title = {Polygons, Polyominoes and Polycubes}}

@article{kiraly2015algebraic,
	author = {Kir{\'a}ly, Franz J and Theran, Louis and Tomioka, Ryota},
	journal = {J. Mach. Learn. Res.},
	number = {1},
	pages = {1391--1436},
	title = {The algebraic combinatorial approach for low-rank matrix completion.},
	volume = {16},
	year = {2015}}

\end{document}